\documentclass[11pt]{article}
\usepackage{fullpage}
\usepackage{algorithmic}
\usepackage{epsfig}
\usepackage{epstopdf}
\usepackage{graphicx}
\usepackage{latexsym}
\usepackage{amsmath}
\usepackage{amsfonts}
\usepackage{amssymb}

\usepackage{mathrsfs}
\usepackage{float}
\usepackage{pifont}
\usepackage{yhmath}
\usepackage{bbm}
\usepackage{eufrak}
\usepackage{amsthm}
\usepackage{booktabs}
\usepackage[usenames,dvipsnames]{color}
\usepackage{stmaryrd}
\usepackage{wasysym}
\usepackage{array}
\usepackage[ruled]{algorithm2e}

\usepackage{bigstrut,multirow,rotating}
\usepackage{cleveref}

\newtheorem{theorem}{Theorem}[section]
\newtheorem{lemma}[theorem]{Lemma}

\newtheorem{corollary}[theorem]{Corollary}
\newtheorem{proposition}[theorem]{Proposition}
\newtheorem{definition}[theorem]{Definition}

\theoremstyle{remark}

\makeatletter
\newcommand\figcaption{\def\@captype{figure}\caption}
\newcommand\tabcaption{\def\@captype{table}\caption}
\makeatother

\newcommand{\bluecomment}[1]{\textcolor{blue}{\textrm{#1}}}

\begin{document}

\title{\bf A $1/2$-Approximation for Budgeted $k$-Submodular Maximization
}

\date{}
\maketitle

\vspace{-3em}
\begin{center}

\author{Chenhao Wang\\
${}$\\
 Beijing Normal University-Zhuhai\\
  Beijing Normal-Hong Kong Baptist University}
\end{center}
\vspace{1em}


\maketitle

\begin{abstract}
A $k$-submodular function naturally generalizes submodular functions by taking as input $k$ disjoint subsets, rather than a single subset. 
 Unlike standard submodular maximization, which only requires selecting elements for the solution, $k$-submodular maximization adds the challenge of determining the subset to which each selected element belongs.
Prior research has shown that the greedy algorithm is a $\frac12$-approximation for the monotone $k$-submodular maximization problem under cardinality or matroid constraints. However, whether a firm $\frac12$-approximation exists for the \emph{budgeted} version (i.e., with a knapsack constraint) has remained open for several years.  We resolve this question affirmatively by proving that the \textsc{1-Guess Greedy} algorithm—which first guesses an appropriate element from an optimal solution before proceeding with the greedy algorithm—achieves a $\frac12$-approximation.
This result is asymptotically tight as $(\frac{k+1}{2k}+\epsilon)$-approximation requires exponentially many value oracle queries even without constraints (Iwata et al., SODA 2016). We further show that \textsc{1-Guess Greedy} is $\frac13$-approximation for the non-monotone problem. 
This algorithm is both simple and parallelizable, making it well-suited for practical applications. Using the thresholding technique from (Badanidiyuru \& Vondr{\'a}k, SODA 2014), it runs in nearly $\tilde O(n^2k^2)$ time. 

The proof idea is simple: we introduce a novel continuous transformation from an optimal solution to a greedy solution, using the multilinear extension to evaluate every fractional solution during the transformation. This continuous analysis approach  yields two key extensions. First, it enables improved approximation ratios of various existing algorithms  for the budgeted setting, including several greedy and streaming algorithm variants. For example, we show that the classic algorithm of returning the better one between the greedy solution and the best singleton achieves $\frac13$-approximation in the monotone case and $\frac14$-approximation in the non-monotone case. Second, our method naturally extends to $k$-submodular maximization problems under broader constraints, offering a more flexible and unified analysis framework.
\end{abstract}


\section{Introduction}\label{sec:intro}

Let $V=\{e_1,\ldots,e_n\}$ be a set of $n$ elements and $2^V$ denote the power set of $V$. A set function $f:2^V\rightarrow \mathbb R$ is \emph{submodular} if  for all sets $X,Y\in 2^V$,
\[f(X)+f(Y)\ge f(X\cup Y)+f(X\cap Y).\]
The submodularity is characterized by the property of \emph{diminishing returns}: for any $X\subseteq Y$ and $a\in V\setminus Y$, $f(X\cup \{a\})-f(X)\ge f(Y\cup \{a\})-f(Y)$. This natural property has established submodular functions as key concepts across numerous domains.

Motivated by Lov{\'a}sz's question \cite{lovasz1983submodular}  concerning generalizations of submodularity that preserve desirable properties such as efficient optimization algorithms, the notion of \emph{$k$-submodular} function was first introduced by Huber and Kolmogorov \cite{huber2012towards}, with earlier related work in \cite{cohen2006complexity}. 
A $k$-submodular function  extends submodularity by considering interactions across $k$ dimensions. Specifically, while a submodular function is defined on single subsets of $V$,  a $k$-submodular function operates on $k$ disjoint subsets of $V$.
Let 
$(k+1)^V:=\{(X_1,\ldots,X_k)~|~X_i\subseteq V ~\forall i\in[k], X_i\cap X_j=\varnothing ~\forall i\neq j\}$ be the family of all $k$-tuples of disjoint subsets, where $[k]:=\{1,\ldots,k\}$.

\begin{definition}\label{def:ksub}
A function $f:(k+1)^V\rightarrow \mathbb R$ is  \emph{$k$-submodular} if and only if  for every $k$-tuples $\mathbf x=(X_1,\ldots,X_k)$ and $\mathbf y=(Y_1,\ldots,Y_k)$ in $(k+1)^V$, 
\[f(\mathbf x)+f(\mathbf y)\ge f(\mathbf x \sqcup\mathbf y)+f(\mathbf x \sqcap\mathbf y),~~\text{where}\]
\[\mathbf x\sqcup\mathbf y:=\Big(X_1\cup Y_1\setminus(\bigcup_{i\neq 1}X_i\cup Y_i),\ldots,X_k\cup Y_k\setminus(\bigcup_{i\neq k}X_i\cup Y_i)\Big),\]
\[\mathbf x\sqcap\mathbf y:=\left(X_1\cap Y_1,\ldots,X_k\cap Y_k\right).\]    
\end{definition}

We write $\mathbf{x} \preceq \mathbf{y}$ if $X_i \subseteq Y_i$ for all $i \in [k]$. For any $\mathbf{x} \preceq \mathbf{y}$, define $\mathbf{y} \setminus \mathbf{x} = (Y_1 \setminus X_1, \ldots, Y_k \setminus X_k)=\mathbf x\sqcap\mathbf y$.  
The function $f$ is called \emph{monotone} if $f(\mathbf{x}) \leq f(\mathbf{y})$ whenever $\mathbf{x} \preceq \mathbf{y}$. We also assume $f(\mathbf{0}) = 0$, where $\mathbf{0} = (\varnothing, \ldots, \varnothing)$. Notably, when $k=1$, $f$ reduces to a submodular function, and when $k=2$, $f$ is known as a bisubmodular function \cite{mccormick2010strongly,huber2014skew}.  


While explicitly given  $k$-submodular functions can be minimized in polynomial time \cite{thapper2012power}, maximizing  $k$-submodular functions is NP-hard, as even the special case of submodular maximization encapsulates many standard NP-hard problems. In contrast to classic submodular maximization, which involves only selecting elements for inclusion, $k$-submodular maximization presents the additional challenge of assigning each element to one of $k$ disjoint subsets.
Over the past decade, there has been increasing theoretical and algorithmic interest in 
$k$-submodular maximization. Theoretically, this problem generalizes several notable combinatorial optimization problems, such as Max $k$-Cut, Max $k$-Partition, and the submodular welfare problem (see \cite{iwata2016improved}). On the practical side, $k$-submodular maximization has found applications in influence maximization, sensor placement, and  feature selection \cite{singh2012bisubmodular}.


Extensive research has focused on developing approximation algorithms for $k$-submodular maximization under support constraints. Formally, given a non-negative $k$-submodular function $f:(k+1)^V\rightarrow \mathbb{R}_{\ge 0}$, the problem is 
\begin{align}  
    &\text{Maximize}~f(\mathbf{x})\\
    &\text{s.t.}~\operatorname{supp}(\mathbf{x})\in \mathcal{P}, \nonumber  
\end{align}  
where $\operatorname{supp}(\mathbf{x}) = \cup_{i=1}^k X_i$ is the \emph{support set} of $\mathbf{x} = (X_1,\ldots,X_k)$, and $\mathcal{P} \subseteq 2^V$ is a prescribed family of subsets of $V$.  
A widely used approach is the greedy algorithm (\textsc{Greedy}, for short), which iteratively adds a feasible element to one of the $k$ subset such that the marginal gain (or density) is maximized. 
Below, we summarize the most notable results from the literature:  

\begin{itemize}  
    \item \textbf{Unconstrained setting.} $\mathcal{P} = 2^V$ is the power set of $V$. For monotone $f$, \textsc{Greedy} achieves a $\frac{1}{2}$-approximation by Ward and {\v{Z}}ivn{\`y} \cite{ward2016maximizing}. Iwata et al.\ \cite{iwata2016improved} present a randomized $\frac{k}{2k-1}$-approximation algorithm and show that any $(\frac{k+1}{2k}+\epsilon)$-approximation requires an exponential number of value oracle queries. For non-monotone $f$, \textsc{Greedy} is $\frac{1}{3}$-approximation \cite{ward2016maximizing}, while Oshima \cite{oshima2021improved} gives a randomized $\frac{k^2+1}{2k^2+1}$-approximation.  
    
    \item \textbf{Cardinality constraint.} $\mathcal{P}$ consists of all subsets of cardinality at most a given $B \in \mathbb{Z}$. \textsc{Greedy} achieves a $\frac{1}{2}$-approximation for monotone $f$ \cite{ohsaka2015monotone} and a $\frac{1}{3}$-approximation for non-monotone $f$ \cite{nguyen2020streaming}.  

    \item \textbf{Matroid constraint.} $(V, \mathcal{P})$ is a matroid of rank $B$, where $\mathcal{P}$ is a family of independent sets. \textsc{Greedy} yields a $\frac{1}{2}$-approximation for monotone $f$ \cite{sakaue2017maximizing} and a $\frac{1}{3}$-approximation for non-monotone $f$ \cite{sun2023maximization}.  

    \item \textbf{Knapsack constraint.} $\mathcal{P} = \{S \subseteq V : c(S) \le B \}$ includes all subsets whose total cost does not exceed $B \in \mathbb{Z}$, with $c(e)$ denoting the cost of $e \in V$ and $c(S) = \sum_{e\in S}c(e)$. Xiao et al.~\cite{xiao2024approximation} show that the \textsc{4-Guess Greedy} algorithm (to be described shortly) achieves a $\frac{1}{2}(1-e^{-2}) \approx 0.432$-approximation for the monotone case, and \textsc{7-Guess Greedy} achieves a $\frac{1}{3}(1-e^{-3}) \approx 0.316$-approximation for the non-monotone case. Recently, Zhou et al.~\cite{zhou2025iclr} propose a  multilinear extension-based continuous algorithm with an approximation ratio of $(\frac{1}{2} - \epsilon)$ for the monotone case and $(\frac{1}{3} - \epsilon)$ for the non-monotone case, and their results apply to multiple knapsack constraints. 
\end{itemize}  

\subsection*{Our contributions}

While monotone $k$-submodular maximization under other support constraints admits a $\frac12$-approximation via the greedy algorithm, whether a $\frac12$-approximation exists for the  $k$-submodular knapsack maximization ($k$SKM) problem has remained unclear. Inspired by the submodular knapsack problem, the focus has been on  the \textsc{$q$-Guess Greedy} algorithms. 
In \textsc{$q$-Guess Greedy} for $k$SKM, all feasible solutions with cardinality at most $q$ are first enumerated;  starting from each, the greedy algorithm is performed by iteratively adding elements and their placements so as to maximize the marginal density (i.e., marginal gain divided by cost), until the budget is exhausted. The algorithm returns the best solution among them. While \textsc{2-Guess Greedy} achieves the optimal $(1-\frac1e)$-approximation  for submodular knapsack \cite{FeldmanNS23,kulik2021refined}, its performance guarantee for $k$SKM has not been tightly analyzed.
Tang et al. \cite{DBLP:journals/orl/TangWC22} first show \textsc{3-Guess Greedy} to be $0.316$-approximation for monotone $k$SKM and later improve it to $0.4$. Xiao et al. \cite{xiao2024approximation} show \textsc{4-Guess Greedy} to be $0.432$-approximation, which is the currently best-known ratio among combinatorial algorithms. 
The non-monotone case is in the same situation but the concern is a  $\frac13$-approximation. 

In this work, we confirm that a single guess suffices to achieve the desired $\frac12-$ and $\frac13-$approximations for $k$SKM, thus matching the results under other support constraints. Moreover, the  $\frac12$-approximation for monotone case is asymptotically tight due to the  $(\frac{k+1}{2k}+\epsilon)$-approximation impossibility result by Iwata et al. \cite{iwata2016improved}.

\begin{theorem}
  \textsc{1-Guess Greedy} is $\frac12$-approximation for monotone $k$SKM and $\frac13$-approximation for non-monotone $k$SKM. It requires $O(n^3k^2)$ queries to the value oracle. 
\end{theorem}

We elaborate on this result from several perspectives.  
First, compared to previous \textsc{$q$-Guess Greedy} studies, our analysis shows that existing bounds in the literature are quite loose. Not only do we improve the approximation ratio, but we also significantly reduce the enumeration size and query complexity. For example, while \textsc{4-Guess Greedy}~\cite{xiao2024approximation} requires $O(n^6 k^5)$ queries, our \textsc{1-Guess Greedy} algorithm only needs $O(n^3 k^2)$.  

Second, the $(\frac{1}{2} - \epsilon)$-approximation continuous algorithm based on the multilinear extension~\cite{zhou2025iclr} incurs a query complexity of $O(k^{\mathrm{poly}(1/\epsilon)} n^{\mathrm{poly}(1/\epsilon)})$, which makes it of theoretical interest only,  even for moderate values of $\epsilon$. In contrast, the combinatorial algorithm \textsc{1-Guess Greedy} is both simple and parallelizable,  and  the $O(n^3k^2)$ query complexity  can be further reduced by a factor of $n$ via the decreasing-threshold technique by Badanidiyuru and Vondr{\'a}k \cite{badanidiyuru2014fast}, with only an $\epsilon$ loss in the approximation ratio, making it more practical.  

Finally, for the monotone submodular knapsack maximization, it is known that \textsc{1-Guess Greedy} achieves a $0.617$-approximation~\cite{FeldmanNS23}. This suggests that the $\frac{1}{2}$-approximation for $k$SKM may not be tight when $k$ is small, and it is natural to ask whether better approximations are possible in this regime. This question remains open and challenging, as no algorithm is known to beat the $\frac{1}{2}$-approximation even for the special case of the cardinality constraint.  

\medskip 
\noindent\textbf{Proof Overview.} For $k$-submodular maximization without constraint or with a cardinality/matroid constraint, to prove the greedy solution is $\frac12$-approximation, the works \cite{ward2016maximizing,ohsaka2015monotone,sakaue2017maximizing} use a discrete transformation from an optimal solution $\mathbf o$ to the greedy solution $\mathbf s$, which replaces one element in $\mathbf o$ by one element in $\mathbf s$ at each time. This transformation makes sense because $|\text{supp}(\mathbf o)|=|\text{supp}(\mathbf s)|$ (note that the maximum independent sets in a matroid have equal size). They prove that at each time the loss of value in this transformation is no more than the gain of the greedy procedure, and summing up over all time gives $f(\mathbf o)-f(\mathbf s)\le f(\mathbf s)-f(\mathbf 0)=f(\mathbf s)$; then the $\frac12$-approximation follows.   However, the discrete transformation does not work well for $k$SKM  because $\mathbf o$ and $\mathbf s$ may contain different number of elements. Previous works on $k$SKM \cite{ha2024improved,pham2022maximizing,xiao2024approximation,DBLP:journals/orl/TangWC22,DBLP:journals/tcs/TangCW24} continue to use this transformation, resulting in bounds that are far from tight.

The key idea of our proof is simple: we introduce a novel continuous transformation instead.  
Let $\mathbf x$ be the continuous process of transforming an optimal solution $\mathbf o$ into  $\mathbf s$, where $\mathbf s$ denotes the greedy solution at the point when the greedy procedure first ``rejects'' an element in $\text{supp}(\mathbf o)$ due to the budget constraint, though that element has the maximum marginal density. Each time $t$ of this continuous process corresponds to a \emph{fractional} solution $\mathbf x(t)$ and a value $F(\mathbf x(t))$, where $F$ is the multilinear extension of $f$. The transforming rate on an element $e$ is $\frac{1}{c(e)}$, and hence, transforming $e$ from one subset to another takes time $c(e)$.  By showing that the negative rate of change of value $F(\mathbf x(t))$ is no more than the rate of change of the greedy value and integrating over all $t$ (from $t=0$ to $t=c(\mathbf o)$), we have 
\[f(\mathbf o)-f(\mathbf s)=F(\mathbf x(0))-F(\mathbf x(c(\mathbf o)))\le f(\mathbf s)-f(\mathbf 0)=f(\mathbf s)\]
when $c(\mathbf o)= c(\mathbf s)$ (see \cref{lem:csco}), and further $f(\mathbf s)\ge \frac{c(\mathbf s)}{2c(\mathbf o)}f(\mathbf o)$ when $c(\mathbf s)\le c(\mathbf o)$ (see \cref{lem:csf}). Finally,  using half the value of the guessed singleton solution to cover the value of the rejected element in $\mathbf o$, we prove the $\frac12$-approximation of \textsc{1-Guess Greedy}. The proof for the $\frac13$-approximation in the non-monotone case is similar.

Our continuous approach leads to two main extensions. First, it enables improved approximation guarantees for a variety of existing algorithms for $k$SKM \cite{ha2024improved,pham2022maximizing,DBLP:journals/tcs/TangCW24}, including several greedy and streaming algorithm variants. For instance, we show that the classic algorithm which returns the better of the greedy solution and the best singleton is $\frac13$-approximation in the monotone case and $\frac14$-approximation in the non-monotone case, improving the previous ratios of 0.273 and 0.219 \cite{DBLP:journals/tcs/TangCW24}. 
Second, our approach extends naturally to $k$-submodular maximization under other support constraints. It provides a more flexible and unified framework for analysis compared to the discrete transformation, as it can handle situations where the solutions of interest have different sizes.

\smallskip

This paper is organized as follows. In  \cref{sec:pre} we present preliminaries. In \cref{sec:mono} and \cref{sec:non} we consider monotone $k$SKM and non-monotone $k$SKM, respectively. In \cref{sec:dis} we discuss the extensions of our proof method. 

\section{Preliminaries}\label{sec:pre}


It is helpful to understand the characteristic properties of $k$-submodular functions. To illustrate these, for any  \( j \) with \( 1 \leq j \leq k \),  define the marginal gain of adding element $e\in V\setminus \text{supp}(\mathbf x)$ to the $j$-th component
of $\mathbf x\in (k + 1)^V$  to be 
\begin{align*}
    &\Delta_{e,j}(\mathbf x) = f(X_1, \ldots, X_{j-1}, X_j \cup \{ e \}, X_{j+1}, \ldots, X_k) - f(X_1, \ldots, X_k).   
\end{align*} 
Ward and {\v{Z}}ivn{\`y} \cite{ward2014maximizing} characterize
the $k$-submodularity by the following two properties: 
\begin{itemize}
    \item orthant submodularity: $
\Delta_{e,j}(\mathbf x) \geq \Delta_{e,j}(\mathbf y)$,  for all $\mathbf x\preceq \mathbf y$ 
and $e \notin \text{supp}(\mathbf y)$;
 \item pairwise monotonicity: $\Delta_{e,j}(\mathbf x) + \Delta_{e,j'}(\mathbf x) \geq 0$, 
for all $\mathbf x\in(k+1)^V$, $e \notin \text{supp}(\mathbf x)$, and $j,j' \in [k]$ with $j \neq j'$. 
\end{itemize}

\begin{proposition}[\cite{ward2014maximizing}]\label{prop:cha}
    A function $f:(k+1)^V\to \mathbb R$ is $k$-submodular if and only if $f$ is
orthant submodular and pairwise monotone.
\end{proposition}



It is often convenient to associate $\mathbf x=(X_1,\ldots,X_k)\in (k+1)^V$ with a vector $\mathbf x\in \{0,1,\ldots,k\}^n$, where $\mathbf x_i=j\in[k]$ if $e_i\in X_j$,   and $\mathbf x_i=0$ if $e_i\notin \text{supp}(\mathbf x)$. We use both notations interchangeably whenever the context is clear.  


\subsection{$k$-Multilinear extension}

The multilinear extension is a key concept in submodular maximization, providing a continuous function that returns the expected value of a submodular function when each element is included independently according to its corresponding probability. Algorithms based on multilinear extension often achieve better approximation ratios for certain submodular problems compared to combinatorial approaches \cite{ChekuriVZ10,DBLP:conf/stoc/Vondrak08}. We consider the multilinear extension of $k$-submodular functions and refer readers to \cite{zhou2025iclr} for more details. Notably, the multilinear extension is used solely for analysis purposes, while our algorithms are purely combinatorial. 

Define $\Delta_k = \left\{ \mathbf{p} \in [0,1]^k : \sum_{j=1}^k p_j \leq 1 \right\}$ to represent probability distributions on discrete values $\{0, 1, \dots, k\}$. Specifically, a vector $\mathbf{p} \in \Delta_k$ represents a probability distribution, where each $j \in [k]$ is assigned  probability $p_j$, and $0$ is assigned  probability $1-\sum_{j=1}^kp_j$. The domain of the $k$-multilinear extension of a $k$-submodular function is  the cube $\Delta_k^n := \left\{ \mathbf{x} \in [0,1]^{nk} : \sum_{j=1}^k \mathbf x_{i,j} \leq 1, \, \forall i \in [n] \right\}$.  

\begin{definition}[$k$-multilinear extension]
    Given  a $k$-submodular function $f : (k+1)^V \to \mathbb{R}_{\geq 0}$, the \textit{$k$-multilinear extension}   
$F : \Delta_k^n \to \mathbb{R}_{\geq 0}$ is defined as  
\[F(\mathbf x) = \sum_{\mathbf s \in (k+1)^V} f(\mathbf s) \prod_{i\in [n]:\mathbf s_i \neq 0} \mathbf x_{i,\mathbf s_i} \prod_{i\in [n]:\mathbf s_i= 0} \left( 1 - \sum_{j=1}^{k}\mathbf x_{i,j} \right).   \]
\end{definition}

For every $\mathbf{x} \in \Delta_k^n$, it holds that $F(\mathbf{x}) = \mathbb{E}[f(\mathbf{s})]$ where $\mathbf{s} \in (k+1)^V$ denotes a random vector: for each element $e_i \in V$, $s_i = j$ for $j \in [k]$ with a probability $\mathbf x_{i,j}$, and $s_i = 0$ otherwise, which occurs independently across all elements.


We need some properties of the $k$-multilinear extension $F$ (see proofs in \cite{zhou2025iclr}). First,  if $f$ is monotone, then $F$ is also monotone (i.e., $\frac{\partial F(\mathbf x)}{\partial \mathbf{x}_{i,j}}\ge 0$ for all $e_i\in V,j\in[k]$). Second, $F$ is \emph{pairwise monotone}: for all $e_i\in V$ and $j_1,j_2\in[k]$, \[\frac{\partial F(\mathbf x)}{\partial \mathbf{x}_{i,j_1}} + \frac{\partial F(\mathbf x)}{\partial \mathbf{x}_{i,j_2}} \geq 0.\]
Third,  $F$ is \emph{multilinear}, i.e., the partial derivative is constant when fixing other coordinates:
\begin{align*}  
\frac{\partial F(\mathbf x)}{\partial\mathbf x_{i,j}} =   
& \sum_{\mathbf s \in (k+1)^V, \, \mathbf s_i = j}   
f(\mathbf s) \prod_{t \in [n] \setminus \{i\} :\mathbf s_t \neq 0} \mathbf x_{t,\mathbf s_t}   
\prod_{t \in [n]:\mathbf s_t = 0} \left( 1 - \sum_{l=1}^k \mathbf x_{t,l} \right) \\
& - \sum_{\mathbf s \in (k+1)^V, \,\mathbf  s_i = 0}   
f(\mathbf s) \prod_{t \in [n]:\mathbf s_t \neq 0} \mathbf x_{t,\mathbf s_t}   
\prod_{t \in [n] \setminus \{i\} :\mathbf s_t = 0} \left( 1 - \sum_{l=1}^k\mathbf x_{t,l} \right)\\
=& ~ \mathbb E[f(\mathbf s)|\mathbf s_i=j]-\mathbb E[f(\mathbf s)|\mathbf s_i=0],
\end{align*}  
where both terms are independent of $\mathbf x_{i,j}$. Finally, $F$ has \emph{element-wise non-positive Hessian:} 
\[  
\frac{\partial^2 F(\mathbf x)}{\partial \mathbf x_{i,j} \partial\mathbf x_{i',j'}}  
\begin{cases}  
= 0 & \text{if } i = {i'}, \\
\le 0 & \text{if } i \neq {i'},  
\end{cases}  
\]  where the first case immediately  implies the multilinearity. 
This property further implies
\begin{equation}\label{eq:submo}
    \frac{\partial F(\mathbf y)}{\partial\mathbf x_{i,j}}\ge \frac{\partial F(\mathbf x)}{\partial\mathbf x_{i,j}}, ~\forall \mathbf y\preceq\mathbf x.
\end{equation}

\section{$1/2$-Approximation for Monotone $k$SKM}\label{sec:mono}

In the $k$-submodular knapsack maximization problem,  each element $e\in V$ is associated with a non-negative cost $c(e)\in\mathbb Z$, and the cost of any $\mathbf x\in(k+1)^V$ is denoted by $c(\mathbf x)=\sum_{e\in \text{supp}(\mathbf x)}c(e)$.  Given a budget $B\in\mathbb Z$, the $k$SKM problem asks to 
\begin{align}
    &\text{Maximize}~f(\mathbf x)\label{eq:knap}\\
    &\text{s.t.}~~c(\mathbf x)\le B. \nonumber
\end{align}
We consider the monotone $k$SKM in this section, where $f$ is a monotone $k$-submodular function.

We describe the \textsc{$q$-Guess Greedy} in Algorithm \ref{alg:enugreedy}. Line 1 enumerates all feasible solutions of size less than $q$ and stores the best one, and it takes $O(n^{q-1}k^{q-1})$ queries to the value oracle. In Line 2-12, for every guess $\mathbf y$ of exactly size $q$, as long as there are remaining elements, we iteratively and greedily select a pair of element $e$ and component $j$ such that the marginal density (the ratio of marginal gain to element cost) is maximized, and add it to the solution $\mathbf s$ if the budget allows. Hence, every guess $\mathbf y$ derives a greedy solution, and the algorithm returns the best one.  Note that each guess takes at most $n$ iterations, and each iteration takes at most $nk$ queries. Since there are $O(n^qk^q)$ guesses in total, the query complexity of \textsc{$q$-Guess Greedy} is $O(n^{q+2}k^{q+1})$.

\begin{algorithm}[H]
	\caption{\hspace{-2pt}~{\textsc{$q$-Guess Greedy}}}
	\label{alg:enugreedy}
	\begin{algorithmic}[1]
	\REQUIRE $V$, $f$, $c$, $B$, enumeration size $q\in\mathbb Z$
	\ENSURE  $\mathbf s^{(g)}\in(k+1)^V$
	\STATE    $\mathbf s^{(g)}\leftarrow\arg\max_{\mathbf x}f(\mathbf x)$ subject to $|\text{supp}(\mathbf x)|<q$ and $c(\mathbf x)\le B$.
	\FOR{ every $\mathbf y$ with $|\text{supp}(\mathbf y)|=q$ and $c(\mathbf y)\le B$}
    \STATE $V^0\leftarrow V\setminus \text{supp}(\mathbf y)$, $\mathbf s\leftarrow \mathbf y$
    \WHILE{  $V^0\neq \varnothing$ }
    \STATE $(e,j)\leftarrow \arg\max\limits_{e\in V^0,j\in[k]}\frac{\Delta_{e,j}(\mathbf s)}{c(e)}$
    \IF {$c(\mathbf s)+c(e)\le B$}
    \STATE $\mathbf s_e\leftarrow j$
    \ENDIF
    \STATE $V^0\leftarrow V^0\setminus\{e\}$
    \ENDWHILE
\STATE $\mathbf s^{(g)}\leftarrow \mathbf s$ if $f(\mathbf s)>f(\mathbf s^{(g)})$.
	\ENDFOR
\RETURN $\mathbf s^{(g)}$
	\end{algorithmic}
\end{algorithm}

When $q=0$, we refer to \textsc{$q$-Guess Greedy} simply as the greedy algorithm, or \textsc{Greedy} for short. It is well known that \textsc{Greedy} does not guarantee bounded performance for knapsack problems, even if the objective is linear. Nevertheless, the next two lemmas show that if a greedy solution incurs a large enough cost before it first rejects any element in the optimal solution, it still achieves good performance.

Formally, let $\mathbf o$ be any feasible (not necessarily optimal) solution of \eqref{eq:knap}. Consider running \textsc{Greedy}, and let $\mathbf s$ be the (partial) greedy solution at the point when the algorithm ``rejects'' an element $e'\in\text{supp}(\mathbf o)$ due to the budget constraint for the first time. That is, 
\[(e',j')\in \arg\max\limits_{e\in V^0,j\in[k]}\frac{\Delta_{e,j}(\mathbf s)}{c(e)}~\text{~~and~~}~c(\mathbf s)+c(e')> B.\]

 We first show that $2f(\mathbf s)\ge f(\mathbf o)$ when $\mathbf s$ and $\mathbf o$ have exactly the same cost. The proof constructs a continuous transformation from $\mathbf o$ to $\mathbf s$, where each time $t\in[0,c(\mathbf o)]$ in this process corresponds to a \emph{fractional} solution $\mathbf x(t)\in\Delta_k^n$ with value $F(\mathbf x(t))$. 
 By showing that for any time $t$, the negative rate of change of value $F(\mathbf x(t))$ never exceeds the rate at which the greedy solution $\mathbf s$'s value increases, and integrating this relationship from $t=0$ to $t=c(\mathbf o)$, we obtain $f(\mathbf o)-f(\mathbf s)=F(\mathbf x(0))-F(\mathbf x(c(\mathbf o)))\le f(\mathbf s)-f(\mathbf 0)=f(\mathbf s)$.

\begin{lemma}[Key lemma]\label{lem:csco}
    If $c(\mathbf s)=c(\mathbf o)$, then $2f(\mathbf s)\ge f(\mathbf o)$.
\end{lemma}
\begin{proof}
    Let vector $\mathbf s(t)\in\Delta_k^n$ denote the (fractional) solution generated by the greedy process at time $t\in[0,c(\mathbf o)]$.   The greedy process starts from $\mathbf s(0)=\mathbf 0$ and ends at $\mathbf s(c(\mathbf o))=\mathbf s$ (we regard $\mathbf s$ as a vector in $\Delta_k^n$ with some abuse of notations). When the pair $(e_i,j)$ is selected by the greedy algorithm,  the corresponding entry $\mathbf s_{ij}(t)$ increases at rate $\frac{1}{c(e_i)}$. Therefore, it takes  time $c(e_i)$ for $\mathbf s_{ij}(t)$ to increase from 0 to 1. Once an entry reaches 1,  the greedy process proceeds to the next pair, repeating this operation until the total time $c(\mathbf o)$  is exhausted.
    
    Let $F'(\mathbf s(t))$ be the derivative with respect to $t$. Note that $F'(\mathbf s(t))$ may not be defined at the breakpoints—those times when the process switches from one element to another—but this does not affect our analysis, as we consider the derivative only at non-breakpoints. By the fundamental theorem of calculus, 
    \begin{equation}\label{eq:basic}
        f(\mathbf s)-f(\mathbf 0)=F(\mathbf s(c(\mathbf o)))-F(\mathbf s(0))=\int_0^{c(\mathbf o)}F'(\mathbf s(t))dt. 
    \end{equation}

   Next, define another vector $\mathbf o(t)\in\Delta_k^n$ for $t\in[0,c(\mathbf o)]$, such that $\mathbf o(0)=\mathbf o$ and $\mathbf o(c(\mathbf o))=\mathbf{0}$. This process represents gradually decreasing the entries of $\mathbf o$ to zero, one element at a time. For each pair $(e_i,j)$ in $\mathbf o$,  $\mathbf o_{ij}(t)$ decreases at rate $\frac{1}{c(e_i)}$. Once it reaches zero, the process proceeds to the next pair in $\mathbf o$, repeating this operation until the total time $c(\mathbf o)$ is exhausted.
The only restriction on the order in which elements in $\mathbf o$ are considered is that, the time when an element $e_i$
  starts decreasing in $\mathbf o(t)$ must be no later than the time it starts increasing in $\mathbf s(t)$, even if $e_i$ belongs to different components in $\mathbf o(t)$ and $\mathbf s(t)$).
This ensures that the combined vector  $\mathbf o(t)+\mathbf s(t)$ remains in $\Delta_k^n$ for all $t$.

  We can now formalize the continuous transformation from  $\mathbf o$ to $\mathbf s$ by setting 
  \begin{equation}\label{eq:aux}
      \mathbf x(t)=\mathbf o(t)+\mathbf s(t)
  \end{equation} for $t\in[0,c(\mathbf o)]$. At the endpoints, $\mathbf x(0)=\mathbf o$ and $\mathbf x(c(\mathbf o))=\mathbf s$, so
    \begin{equation}\label{eq:yui}
        f(\mathbf s)-f(\mathbf o)=F(\mathbf x(c(\mathbf o)))-F(\mathbf x(0))=\int_0^{c(\mathbf o)}F'(\mathbf x(t))dt.
    \end{equation}
     Moreover, define $\mathbf y(t)=\mathbf o(t)+\lfloor\mathbf s\rfloor(t)$, where $\lfloor\mathbf s\rfloor(t)$ denotes the vector obtained by rounding every non-one entry of $\mathbf s(t)$ down to zero.  In this construction, while $\mathbf o(t)$ changes continuously, $\lfloor\mathbf s\rfloor(t)$ changes discretely: $\lfloor\mathbf s\rfloor_{ij}(t)$ jumps from 0 to 1 at the moment when $(e_i,j)$ has been considered for a total time $c(e_i)$. As $\mathbf y(0)=\mathbf o$ and $\mathbf y(c(\mathbf o))=\mathbf s$, $\mathbf y(t)$ also represents a transformation from $\mathbf o$ to $\mathbf s$.

   Note that $\mathbf y(t)\preceq \mathbf x(t)$, and they differ in at most one entry at any time $t$.  We claim that for all non-breakpoint $t$,  
     \begin{equation}\label{eq:tosee}
        F'(\mathbf x(t))\ge F'(\mathbf y(t))\le 0.
     \end{equation}
To see \eqref{eq:tosee}, consider the updates at time $t$: suppose $\mathbf x_{ij}(t)$ and $\mathbf y_{ij}(t)$ is decreasing, and $\mathbf x_{i'j'}(t)$ is increasing, with other entries unchanged. Then
\begin{align*}
    &F'(\mathbf x(t))=\nabla F\cdot \mathbf x'(t)=\frac{\partial F(\mathbf x(t))}{\partial \mathbf x_{i'j'}}\frac{1}{c(e_{i'})}-\frac{\partial F(\mathbf x(t))}{\partial \mathbf x_{ij}}\frac{1}{c(e_i)}~\text{~and~}~\\
    &F'(\mathbf y(t))=\nabla F\cdot \mathbf y'(t)=-\frac{\partial F(\mathbf y(t))}{\partial \mathbf x_{ij}}\frac{1}{c(e_i)}.
\end{align*}  
Since $\mathbf y(t)\preceq \mathbf x(t)$, by the property \eqref{eq:submo}, we have 
\[\frac{\partial F(\mathbf y(t))}{\partial \mathbf x_{ij}}\frac{1}{c(e_i)}\ge \frac{\partial F(\mathbf x(t))}{\partial \mathbf x_{ij}}\frac{1}{c(e_i)}.\]
Further by the monotonicity of $F$, we obtain \eqref{eq:tosee}. 
     
  Now we show that the negative rate of change of value $F(\mathbf y(t))$ never exceeds the rate at
which the greedy value $F(\mathbf s(t))$ increases. Assume that, at non-breakpoint $t$,  $\mathbf y_{ij}(t)$ is decreasing and $\mathbf s_{i'j'}(t)$ is increasing, with other entries unchanged. Then,
     \begin{align*}
         -F'(\mathbf y(t))=\frac{\partial F(\mathbf y(t))}{\partial \mathbf x_{ij}}\frac{1}{c(e_i)}&\le   \frac{\partial F(\lfloor\mathbf s\rfloor(t))}{\partial \mathbf x_{ij}}\frac{1}{c(e_i)}\\
         &\le\frac{\partial F(\lfloor\mathbf s\rfloor(t))}{\partial \mathbf x_{i'j'}}\frac{1}{c(e_{i'})}\\
         &=\frac{\partial F(\mathbf s(t))}{\partial \mathbf x_{i'j'}}\frac{1}{c(e_{i'})}=F'(\mathbf s(t)). 
     \end{align*}
     The first inequality is due to the fact $\lfloor\mathbf s\rfloor(t)\preceq \mathbf y(t)$ and the property \eqref{eq:submo}.  The second inequality follows because the greedy algorithm always selects the pair $(e_{i'},j')$ with the maximum marginal density, and $\frac{\partial F(\lfloor\mathbf s\rfloor(t))}{\partial \mathbf x_{i'j'}}$ is the marginal gain of $(e_{i'},j')$ by the definition of derivative. The second last equality is due to the multilinearity. 
   
     Therefore, combining the above with \eqref{eq:basic}, \eqref{eq:yui} and \eqref{eq:tosee}, we obtain
     \begin{align*}
         f(\mathbf o)-f(\mathbf s)=\int_0^{c(\mathbf o)}-F'(\mathbf x(t))dt&\le \int_0^{c(\mathbf o)}-F'(\mathbf y(t))dt\\
         &\le \int_0^{c(\mathbf o)}F'(\mathbf s(t))dt\\
         &=f(\mathbf s),
     \end{align*}
     establishing the proof. 
\end{proof}

We note that, in addition to our new approach of constructing an auxiliary point $\mathbf x(t) =\mathbf s(t) +\mathbf o(t)$ in \eqref{eq:aux} for analyzing the combinatorial algorithm, there are some other constructions in existing analysis of continuous algorithms. Zhou et al.~\cite{zhou2025iclr} use the auxiliary point $\mathbf s(t) + (1 - t)\mathbf o$ to prove a $(\frac12-\epsilon)$-approximation of the Frank-Wolfe-type continuous greedy algorithm for a very general class of constraints (including multiple knapsack constraints). Also in  an early version of \cite{zhou2025iclr}, there is a more similar construction of auxiliary point to ours for analyzing the continuous greedy algorithm under a single knapsack constraint. 


The bound in Lemma \ref{lem:csco}  can be extended to arbitrary cost of $\mathbf s$. 

\begin{lemma}\label{lem:csf}
    If $c(\mathbf s)= \beta\cdot c(\mathbf o)$, then $f(\mathbf s)\ge \min\{\frac{\beta}{2},\frac12\}\cdot f(\mathbf o)$.
\end{lemma}
\begin{proof}
If $\beta\le 1$, let $s_l$ be the last element added to $\mathbf s$. We construct a new solution $\mathbf s'$ by replacing $s_l$ in $\mathbf s$ with an element $s_l'$ such that its cost is $c(s_l')=c(s_l)+c(\mathbf o)- c(\mathbf s)$, while preserving all marginal densities. The total cost now satisfies  $c(\mathbf s')=c(\mathbf o)$. Because the elements selected by the greedy algorithm has non-increasing marginal densities, we have $\frac{f(\mathbf s)}{c(\mathbf s)}\ge \frac{f(\mathbf s')}{c(\mathbf s')}$, and thus $f(\mathbf s)\ge \beta\cdot f(\mathbf s')$. Then, applying Lemma \ref{lem:csco}, we obtain $2f(\mathbf s')\ge f(\mathbf o)$, which leads to $f(\mathbf s)\ge \frac{\beta}{2}f(\mathbf o)$. 

If $\beta>1$,  we construct a new solution $\mathbf o'$ by adding some zero-value elements to $\mathbf o$ such that $c(\mathbf o')=c(\mathbf s)$ and $f(\mathbf o)=f(\mathbf o')$.  
Applying  Lemma \ref{lem:csco} leads to $f(\mathbf s)\ge \frac12f(\mathbf o')=\frac12f(\mathbf o)$. 
\end{proof}

Now we are ready to show that the greedy algorithm equipped with a single guess can achieve a $\frac12$-approximation. 

\begin{theorem}\label{thm:12}
    \textsc{$1$-Guess Greedy} is $\frac12$-approximation for monotone $k$-submodular knapsack, in $O(n^3k^2)$ time. 
\end{theorem}
\begin{proof}
For any $\mathbf x\in(k+1)^V$, define the function $f_{\mathbf x}:(k+1)^{V\setminus\text{supp}(\mathbf x)}\to\mathbb R$ by $f_{\mathbf x}(\mathbf x')=f(\mathbf x'\sqcup\mathbf x)-f(\mathbf x)$. It is clear that $f_{\mathbf x}$ is also $k$-submodular.

    Let $\mathbf o^*$ be an optimal solution of $k$SKM. Let $\mathbf y\in\arg\max_{\mathbf x\preceq \mathbf o^*:|\text{supp}(\mathbf x)|=1}f(\mathbf x)$ be the singleton solution corresponding to the element in $\mathbf o^*$  that achieves the largest value. Let  $\mathbf z\in\arg\max_{\mathbf x\preceq \mathbf o^*\setminus \mathbf y:|\text{supp}(\mathbf x)|=1}c(\mathbf x)$ be the singleton solution corresponding to the remaining element in $\mathbf o^*\setminus \mathbf y$ with the largest cost. Let $\mathbf r=\mathbf o^*\setminus\mathbf y\setminus\mathbf z$ denote the collection of the remaining elements.   We have 
    \[f(\mathbf o^*)=f(\mathbf y)+f_{\mathbf y}(\mathbf r)+f_{\mathbf y\sqcup \mathbf r}(\mathbf z).\] 
    
    Suppose that \textsc{$1$-Guess Greedy} guesses $\mathbf y$ correctly, and let $\mathbf s$ be the greedy solution at the point when the algorithm first rejects an element $e\in\mathbf o^*$ due to the budget constraint. By construction, $c(\mathbf s)+c(e)>B\ge c(\mathbf o^*)$, which leads to 
    $c(\mathbf s\setminus \mathbf y)\ge c(\mathbf r)$ by the definition of $\mathbf r$. Applying Lemma \ref{lem:csf} to  $\mathbf s \setminus \mathbf y$, $\mathbf r$ and $f_{\mathbf y}$, we obtain 
    \[f_{\mathbf y}(\mathbf s\setminus \mathbf y)\ge \frac12 f_{\mathbf y}(\mathbf r).\]
    Therefore, it follows that \begin{align*}
        f(\mathbf s)&=f(\mathbf y)+f_{\mathbf y}(\mathbf s\setminus \mathbf y)
        \ge \frac12 f(\mathbf y)+\frac12 f_{\mathbf y\sqcup \mathbf r}(\mathbf z)+\frac12 f_{\mathbf y}(\mathbf r)=\frac12 f(\mathbf o^*),
    \end{align*}
    where the inequality is due to the definition of $\mathbf y$.
\end{proof}

The $O(n^3k^2)$ query complexity of \textsc{1-Guess Greedy} can be further reduced by a factor of $n$ using the decreasing-threshold technique from \cite{badanidiyuru2014fast}, at the cost of an additional $\epsilon$ loss in the approximation ratio. Specifically, there exists a $(\frac12-\epsilon)$-approximation algorithm for monotone $k$SKM within $\tilde O(n^2k^2)$ queries, where the $\tilde O$ notation disregards poly-logarithmic terms. 

\section{$1/3$-Approximation for Non-monotone $k$SKM}\label{sec:non}

For non-monotone $k$-submodular functions, we only have pairwise monotonicity rather than monotonicity. Nevertheless, the analysis of the $\frac13$-approximation  proceeds in a way analogous to the monotone case in Section~\ref{sec:mono}.

Let  $\mathbf o$ be an arbitrary feasible solution of the non-monotone $k$SKM. Consider running \textsc{Greedy}, and let $\mathbf s$ be the greedy solution at the point when the algorithm rejects an element in $\text{supp}(\mathbf o)$ due to the budget constraint for the first time. 

\begin{lemma}\label{lem:csconon}
    If $c(\mathbf s)=c(\mathbf o)$, then $3f(\mathbf s)\ge f(\mathbf o)$.
\end{lemma}
\begin{proof}
The proof follows the same structure as Lemma \ref{lem:csco} for monotone $k$SKM; after defining the continuous processes $\mathbf s(t)$ and $\mathbf x(t)$ in the same way, we need only replace \eqref{eq:tosee} by 
     \begin{equation}\label{eq:toseenon}
         F'(\mathbf x(t))\ge -2F'(\mathbf s(t)).
     \end{equation}
To show \eqref{eq:toseenon}, consider the updates at time $t$: suppose $\mathbf x_{ij}(t)$ is decreasing, and $\mathbf x_{i'j'}(t)$ is increasing, with other entries unchanged. 
Let $l\in[k]$ and $l\neq j'$. By the pairwise monotonicity, we have
\begin{equation}\label{eq:com0}
    \frac{\partial F(\mathbf x(t))}{\partial \mathbf x_{i'j'}}+\frac{\partial F(\mathbf x(t))}{\partial \mathbf x_{i'l}}\ge 0.
\end{equation} 
By the property \eqref{eq:submo} of $F$ and the greedy selection of the algorithm, we have
\begin{align}
   & \frac{\partial F(\mathbf x(t))}{\partial \mathbf x_{i'l}}\frac{1}{c(e_{i'})} \le \frac{\partial F(\lfloor\mathbf s\rfloor(t))}{\partial \mathbf x_{i'l}}\frac{1}{c(e_{i'})} \le F'(\mathbf s(t))~\text{and}\label{eq:com1} \\
   & \frac{\partial F(\mathbf x(t))}{\partial \mathbf x_{ij}}\frac{1}{c(e_i)}\le \frac{\partial F(\lfloor\mathbf s\rfloor(t))}{\partial \mathbf x_{ij}}\frac{1}{c(e_i)}\le  F'(\mathbf s(t)).\label{eq:com2}
\end{align}
Combining \eqref{eq:com0}, \eqref{eq:com1} and \eqref{eq:com2} gives
\begin{align}
    F'(\mathbf x(t))&=\frac{\partial F(\mathbf x(t))}{\partial \mathbf x_{i'j'}}\frac{1}{c(e_{i'})}-\frac{\partial F(\mathbf x(t))}{\partial \mathbf x_{ij}}\frac{1}{c(e_i)}\nonumber\\
   &\ge -\frac{\partial F(\mathbf x(t))}{\partial \mathbf x_{i'l}}\frac{1}{c(e_{i'})}  -F'(\mathbf s(t))  \nonumber\\
&\ge -2 F'(\mathbf s(t)).\nonumber
   \end{align}
     Therefore, \eqref{eq:toseenon} holds for all $t$. It follows that
     \begin{align*}
         f(\mathbf o)-f(\mathbf s)&= \int_0^{c(\mathbf o)}-F'(\mathbf x(t))dt
         \le 2\int_0^{c(\mathbf o)}F'(\mathbf s(t))dt
         =2f(\mathbf s),
     \end{align*}
     establishing the proof. 
\end{proof}

The following lemma generalizes the bound in Lemma \ref{lem:csconon} to the case when $\mathbf s$ has arbitrary cost. 

\begin{lemma}\label{lem:csfff}
    If $c(\mathbf s)= \beta\cdot c(\mathbf o)$, then $f(\mathbf s)\ge \min\{\frac{\beta}{3},\frac13\}\cdot f(\mathbf o)$.
\end{lemma}
\begin{proof}
If $\beta\le 1$, let $s_l$ be the last element added to $\mathbf s$. We replace $s_l$ in $\mathbf s$ by an element $s_l'$ with cost $c(s_l')=c(s_l)+c(\mathbf o)- c(\mathbf s)$, while maintaining all marginal densities. The resulting solution is denoted by $\mathbf s'$. It is clear that $c(\mathbf s')=c(\mathbf o)$ and $f(\mathbf s)\ge \beta\cdot f(\mathbf s')$ by the submodularity. By Lemma \ref{lem:csconon} we have $3f(\mathbf s')\ge f(\mathbf o)$. Therefore, we have $f(\mathbf s)\ge \frac{\beta}{3}f(\mathbf o)$. 

If $\beta>1$, consider a partial solution $\mathbf s'$ of $\mathbf s$ so that $c(\mathbf s')=c(\mathbf o)$, i.e., the greedy procedure stops at time $c(\mathbf o)$. Applying  Lemma \ref{lem:csconon} gets $f(\mathbf s)\ge F(\mathbf s')\ge \frac{1}{3}f(\mathbf o)$. 
\end{proof}

\begin{theorem}
    \textsc{$1$-Guess Greedy} is $\frac13$-approximation for non-monotone $k$SKM. 
\end{theorem}
\begin{proof}
    Let $\mathbf o^*$ be an optimal solution. Define $\mathbf y,\mathbf z,\mathbf r$ as in the proof of Theorem \ref{thm:12}, such that   \[f(\mathbf o^*)=f(\mathbf y)+f_{\mathbf y}(\mathbf r)+f_{\mathbf y\sqcup \mathbf r}(\mathbf z).\] 
    Suppose that \textsc{$1$-Guess Greedy} guesses $\mathbf y$ correctly, and let $\mathbf s$ be the greedy solution at the point when the algorithm first rejects an element $e\in\mathbf o^*$ due to the budget constraint. By construction, $c(\mathbf s)+c(e)>B\ge c(\mathbf o^*)$, which leads to 
    $c(\mathbf s\setminus \mathbf y)\ge c(\mathbf r)$ by the definition of $\mathbf r$. Applying Lemma \ref{lem:csfff} to  $\mathbf s \setminus \mathbf y$, $\mathbf r$ and $f_{\mathbf y}$, we obtain 
    \[f_{\mathbf y}(\mathbf s\setminus \mathbf y)\ge \frac13 f_{\mathbf y}(\mathbf r).\]
    Therefore, it follows that \begin{align*}
        f(\mathbf s)&=f(\mathbf y)+f_{\mathbf y}(\mathbf s\setminus \mathbf y)
        \ge \frac12 f(\mathbf y)+\frac12 f_{\mathbf y\sqcup \mathbf r}(\mathbf z)+\frac13 f_{\mathbf y}(\mathbf r)\ge\frac13 f(\mathbf o^*).
    \end{align*}
\end{proof}

\section{Extensions and Discussions}\label{sec:dis}

The extensions of our analysis approach based on continuous transformations are twofold. First, our method can be used to improve the approximation ratios of various existing $k$SKM algorithms beyond \textsc{$q$-Guess Greedy}, including several greedy and streaming algorithm variants \cite{ha2024improved,pham2022maximizing,DBLP:journals/tcs/TangCW24}. Second, the method naturally extends to $k$-submodular maximization under broader support constraints, offering a more general and flexible analysis framework.

\paragraph*{Improving other existing $k$SKM algorithms}


To illustrate how our method enhances current approaches for $k$SKM, consider the classic \textsc{Greedy+Singleton} algorithm, which selects the better of the standard greedy solution and the best feasible singleton. This algorithm is well-known, with a long history, and achieves a $\frac{1}{2}$-approximation for linear objectives under knapsack constraints, while its approximation ratio for monotone submodular objectives lies in $[0.427, 0.42945]$ (see \cite{FeldmanNS23,DBLP:journals/orl/KulikSS21}). For $k$SKM, prior work \cite{DBLP:journals/tcs/TangCW24} establishes approximation ratios of $0.273$ in the monotone case and $0.219$ in the non-monotone case. Using the bounds obtained from our continuous method, we are able to improve these ratios to $\frac{1}{3}$ and $\frac{1}{4}$, respectively.
 
\begin{proposition}
    \textsc{Greedy+Singleton} is $\frac13$-approximation for monotone $k$SKM and $\frac14$-approximation for non-monotone $k$SKM, in $O(n^2k)$ time. 
\end{proposition}
\begin{proof}
    Let $\mathbf o^*$ be an optimal solution, and $\mathbf z\preceq \mathbf o^*$ be the singleton in $\mathbf o^*$ with the largest cost. Denote by $\mathbf s$ the greedy solution when it first rejects an element in $\mathbf o^*$ due to the budget constraint, and it implies that $c(\mathbf s)\ge c(\mathbf o^*\setminus \mathbf z)$.
    When $f$ is monotone,
    if $f(\mathbf z)\ge \frac13f(\mathbf o^*)$, then \textsc{Singleton} finds it and gives the approximation ratio. Otherwise, $f(\mathbf o^*\setminus \mathbf z)\ge \frac23f(\mathbf o^*)$. By Lemma \ref{lem:csf}, we have $f(\mathbf s)\ge \frac12f(\mathbf o^*\setminus \mathbf z)\ge \frac13f(\mathbf o^*)$.

When $f$ is non-monotone, if $f(\mathbf z)\ge \frac14f(\mathbf o^*)$, then \textsc{Singleton} finds it and gives the approximation ratio. Otherwise, $f(\mathbf o^*\setminus \mathbf z)\ge \frac34f(\mathbf o^*)$. By Lemma \ref{lem:csfff}, we have $f(\mathbf s)\ge \frac13f(\mathbf o\setminus \mathbf z)\ge \frac14f(\mathbf o^*)$.
\end{proof}

The \textsc{Greedy$^+$} algorithm follows the standard \textsc{Greedy} procedure, but outputs the best feasible solution obtained by augmenting any partial greedy solution with a single additional element. Previously, it was shown to achieve a $\frac13$-approximation for monotone $k$SKM and a $\frac14$-approximation for non-monotone $k$SKM \cite{ha2024improved,DBLP:conf/cocoa/TangCWWJ23}. By constructing a continuous transformation that differs from that in the proof of Lemma \ref{lem:csco}, we can improve these ratios to $\frac{4}{11}\approx 0.363$ and $\frac{5}{18}\approx 0.277$, respectively. Furthermore,  existing (semi-)streaming algorithms provide a $(\frac14-\epsilon)$-approximation for monotone $k$SKM \cite{ha2024improved,pham2022maximizing}. With our continuous transformation method, we can enhance their approximation ratio to at least $0.26-\epsilon$ by adding one additional pass of input stream, while maintaining the $O(nk/\epsilon)$ query complexity and $O(\min\{n,B\}/\epsilon)$ space complexity.

\paragraph*{$k$-submodular maximization  under other constraints}
For $k$-submodular maximization problems under a cardinality constraint, a matroid constraint, or without any constraint, both the optimal solution $\mathbf o$ and the greedy solution 
$\mathbf s$ have the same cardinality, denoted by $B$ (as all maximum independent sets in a matroid have equal size).  Previous works~\cite{ohsaka2015monotone, sakaue2017maximizing, ward2016maximizing} use a discrete transformation to show that the greedy algorithm is $\frac12$-approximation in the monotone case.
 Let $\mathbf s^{(j)}$ be the greedy solution after the $j$-th iteration, with $\mathbf s^{(0)}=\mathbf 0$ and $\mathbf s^{(B)}=\mathbf s$. The transformation starts with  \(\mathbf o^{(0)}=\mathbf{o}\) and, through a sequence of updates \(\mathbf{o}^{(0)}, \mathbf{o}^{(1)}, \ldots, \mathbf{o}^{(B)}\), gradually adjusts it to match the greedy solution \(\mathbf{s}\). At each step \(j\), the element and its placement in \(\mathbf{o}^{(j-1)}\) corresponding to the greedy choice at that iteration is replaced with that chosen by the greedy algorithm in \(\mathbf{s}^{(j)}\). Repeating this process for all \(B\) steps, \(\mathbf{o}^{(B)}\) becomes identical to \(\mathbf{s}^{(B)} = \mathbf{s}\).  To show the $\frac12$-approximation, it suffices to prove $f(\mathbf s^{(j)})-f(\mathbf s^{(j-1)})\ge f(\mathbf o^{(j-1)})-f(\mathbf o^{(j)})$ for all $j\in[B]$, because summing up all inequalities gives $f(\mathbf s)-f(\mathbf 0)\ge f(\mathbf o)-f(\mathbf s)$.
The proof of the $\frac13$-approximation for the greedy algorithm in the non-monotone case follows similarly.

This discrete transformation is, in fact, a special case of our continuous transformation as used in the proof of Lemma \ref{lem:csco}, where every element has unit transformation rate and processing time, and the total duration is $B$. Hence, our continuous approach naturally applies to problems under these constraints. Moreover, it offers greater flexibility compared to the discrete approach, as it can handle situations where the solutions of interest do not have identical cardinality, and it may inspire new insights for analyzing combinatorial algorithms in such situations.


\bibliography{reference}

\begin{thebibliography}{10}

\bibitem{badanidiyuru2014fast}
Ashwinkumar Badanidiyuru and Jan Vondr{\'a}k.
\newblock Fast algorithms for maximizing submodular functions.
\newblock In {\em Proceedings of the 25th Annual ACM-SIAM Symposium on Discrete Algorithms (SODA)}, pages 1497--1514. SIAM, 2014.

\bibitem{ChekuriVZ10}
Chandra Chekuri, Jan Vondr{\'{a}}k, and Rico Zenklusen.
\newblock Dependent randomized rounding via exchange properties of combinatorial structures.
\newblock In {\em 51th Annual {IEEE} Symposium on Foundations of Computer Science (FOCS)}, pages 575--584, 2010.

\bibitem{cohen2006complexity}
David~A Cohen, Martin~C Cooper, Peter~G Jeavons, and Andrei~A Krokhin.
\newblock The complexity of soft constraint satisfaction.
\newblock {\em Artificial Intelligence}, 170(11):983--1016, 2006.

\bibitem{FeldmanNS23}
Moran Feldman, Zeev Nutov, and Elad Shoham.
\newblock Practical budgeted submodular maximization.
\newblock {\em Algorithmica}, 85(5):1332--1371, 2023.

\bibitem{ha2024improved}
Dung~TK Ha, Canh~V Pham, and Tan~D Tran.
\newblock Improved approximation algorithms for k-submodular maximization under a knapsack constraint.
\newblock {\em Computers \& Operations Research}, 161:106452, 2024.

\bibitem{huber2012towards}
Anna Huber and Vladimir Kolmogorov.
\newblock Towards minimizing $k$-submodular functions.
\newblock In {\em Proceedings of the 2nd International Symposium on Combinatorial Optimization (ISCO)}, pages 451--462. Springer, 2012.

\bibitem{huber2014skew}
Anna Huber, Andrei Krokhin, and Robert Powell.
\newblock Skew bisubmodularity and valued csps.
\newblock {\em SIAM Journal on Computing}, 43(3):1064--1084, 2014.

\bibitem{iwata2016improved}
Satoru Iwata, Shin-ichi Tanigawa, and Yuichi Yoshida.
\newblock Improved approximation algorithms for $k$-submodular function maximization.
\newblock In {\em Proceedings of the 27th Annual ACM-SIAM Symposium on Discrete Algorithms (SODA)}, pages 404--413, 2016.

\bibitem{kulik2021refined}
Ariel Kulik, Roy Schwartz, and Hadas Shachnai.
\newblock A refined analysis of submodular greedy.
\newblock {\em Operations Research Letters}, 49(4):507--514, 2021.

\bibitem{DBLP:journals/orl/KulikSS21}
Ariel Kulik, Roy Schwartz, and Hadas Shachnai.
\newblock A refined analysis of submodular greedy.
\newblock {\em Operations Research Letters}, 49(4):507--514, 2021.

\bibitem{lovasz1983submodular}
L{\'a}szl{\'o} Lov{\'a}sz.
\newblock Submodular functions and convexity.
\newblock {\em Mathematical Programming The State of the Art: Bonn 1982}, pages 235--257, 1983.

\bibitem{mccormick2010strongly}
S~Thomas McCormick and Satoru Fujishige.
\newblock Strongly polynomial and fully combinatorial algorithms for bisubmodular function minimization.
\newblock {\em Mathematical Programming}, 122(1):87--120, 2010.

\bibitem{nguyen2020streaming}
Lan Nguyen and My~T Thai.
\newblock Streaming $k$-submodular maximization under noise subject to size constraint.
\newblock In {\em Proceedings of the 37th International Conference on Machine Learning (ICML)}, pages 7338--7347. PMLR, 2020.

\bibitem{ohsaka2015monotone}
Naoto Ohsaka and Yuichi Yoshida.
\newblock Monotone $k$-submodular function maximization with size constraints.
\newblock In {\em Proceedings of the 28th International Conference on Neural Information Processing Systems (NeurIPS)}, volume~1, pages 694--702, 2015.

\bibitem{oshima2021improved}
Hiroki Oshima.
\newblock Improved randomized algorithm for $k$-submodular function maximization.
\newblock {\em SIAM Journal on Discrete Mathematics}, 35(1):1--22, 2021.

\bibitem{pham2022maximizing}
Canh~V Pham, Quang~C Vu, Dung~KT Ha, Tai~T Nguyen, and Nguyen~D Le.
\newblock Maximizing k-submodular functions under budget constraint: applications and streaming algorithms.
\newblock {\em Journal of Combinatorial Optimization}, 44(1):723--751, 2022.

\bibitem{sakaue2017maximizing}
Shinsaku Sakaue.
\newblock On maximizing a monotone $k$-submodular function subject to a matroid constraint.
\newblock {\em Discrete Optimization}, 23:105--113, 2017.

\bibitem{singh2012bisubmodular}
Ajit Singh, Andrew Guillory, and Jeff Bilmes.
\newblock On bisubmodular maximization.
\newblock In {\em Proceedings of the 15th International Conference on Artificial Intelligence and Statistics (AISTATS)}, pages 1055--1063. PMLR, 2012.

\bibitem{sun2023maximization}
Yunjing Sun, Yuezhu Liu, and Min Li.
\newblock Maximization of k-submodular function with a matroid constraint.
\newblock In {\em Proceedings of the 17th Annual Conference on Theory and Applications of Models of Computation}, pages 1--10. Springer, 2022.

\bibitem{DBLP:journals/tcs/TangCW24}
Zhongzheng Tang, Jingwen Chen, and Chenhao Wang.
\newblock Greedy+singleton: An efficient approximation algorithm for \emph{k}-submodular knapsack maximization.
\newblock {\em Theoretical Computer Science}, 984:114320, 2024.

\bibitem{DBLP:conf/cocoa/TangCWWJ23}
Zhongzheng Tang, Jingwen Chen, Chenhao Wang, Tian Wang, and Weijia Jia.
\newblock Greedy+max: An efficient approximation algorithm for k-submodular knapsack maximization.
\newblock In {\em Combinatorial Optimization and Applications - 17th International Conference (COCOA)}, volume 14461, pages 287--299. Springer, 2023.

\bibitem{DBLP:journals/orl/TangWC22}
Zhongzheng Tang, Chenhao Wang, and Hau Chan.
\newblock On maximizing a monotone \emph{k}-submodular function under a knapsack constraint.
\newblock {\em Operations Research Letters}, 50(1):28--31, 2022.

\bibitem{thapper2012power}
Johan Thapper and Stanislav Zivn{\'{y}}.
\newblock The power of linear programming for valued csps.
\newblock In {\em 53rd Annual {IEEE} Symposium on Foundations of Computer Science (FOCS)}, pages 669--678, 2012.

\bibitem{DBLP:conf/stoc/Vondrak08}
Jan Vondr{\'{a}}k.
\newblock Optimal approximation for the submodular welfare problem in the value oracle model.
\newblock In {\em Proceedings of the 40th Annual {ACM} Symposium on Theory of Computing (STOC)}, pages 67--74, 2008.

\bibitem{ward2014maximizing}
Justin Ward and Stanislav {\v{Z}}ivn{\`y}.
\newblock Maximizing bisubmodular and k-submodular functions.
\newblock In {\em Proceedings of the 25th Annual ACM-SIAM Symposium on Discrete Algorithms (SODA)}, pages 1468--1481. SIAM, 2014.

\bibitem{ward2016maximizing}
Justin Ward and Stanislav {\v{Z}}ivn{\`y}.
\newblock Maximizing $k$-submodular functions and beyond.
\newblock {\em ACM Transactions on Algorithms}, 12(4):1--26, 2016.

\bibitem{xiao2024approximation}
Hao Xiao, Qian Liu, Yang Zhou, and Min Li.
\newblock Approximation algorithms for k-submodular maximization subject to a knapsack constraint.
\newblock {\em Journal of the Operations Research Society of China}, pages 1--16, 2024.

\bibitem{zhou2025iclr}
Huanjian Zhou, Lingxiao Huang, and Baoxiang Wang.
\newblock Improved approximation algorithms for $k$-submodular maximization via multilinear extension.
\newblock {\em To appear in the 13th International Conference on Learning Representations (ICLR)}, 2025.

\end{thebibliography}
\end{document}